\documentclass[11pt, centerh1]{article}

\usepackage{graphicx,colonequals,bm}
\usepackage{natbib}
\usepackage{url}
\usepackage{amsmath,amsthm}
\usepackage{amsfonts}
\usepackage{amssymb}
\usepackage{url}
\usepackage{color}
\usepackage{subfigure}
\usepackage{slashbox}
\usepackage{mathtools}
\def\bx{\boldsymbol{x}}
\def\by{\boldsymbol{y}}
\def\bz{\boldsymbol{z}}

\def\bX{\boldsymbol{X}}
\def\bY{\boldsymbol{Y}}

\def\b0{\boldsymbol{0}}
\def\b1{\boldsymbol{1}}
\def\cC{\mathcal{C}}

\def\cF{\mathcal{F}}
\def\cH{\mathcal{H}}

\def\cX{\mathcal{X}}

\def\cZ{\mathcal{Z}}
\def\mR{\mathbb{R}}
\def\mN{\mathbb{N}}
\def\EE{\mathbb{E}}

\def\bmu{\boldsymbol{\mu}}

\def\bpsi{\boldsymbol{\psi}}

\def\bbeta{\boldsymbol{B}}

\def\bvartheta{\boldsymbol{\theta}}
\def\bvartheta{\boldsymbol{\vartheta}}

\def\bpi{\boldsymbol{\pi}}

\def\bSigma{\boldsymbol{\Sigma}}

\def\bPsi{\boldsymbol{\Psi}}
\def\bvartheta{\boldsymbol{\vartheta}}
\def\bTheta{\boldsymbol{\varTheta}}

\newtheorem{Thm}{Theorem}

\begin{document}

\title{Multivariate Response and Parsimony for Gaussian Cluster-Weighted Models}
\author{Utkarsh J.\ Dang \thanks{Department of Mathematics and Statistics, McMaster University, Hamilton, Ontario, Canada. E-mail: udang@mcmaster.ca.}, Antonio Punzo\thanks{Department of Economics and Business, University of Catania, Catania, Italy. E-mail: antonio.punzo@unict.it.}, Paul D.\ McNicholas\thanks{Department of Mathematics and Statistics, McMaster University, Hamilton, Ontario, Canada. E-mail: mcnicholas@math.mcmaster.ca.}, \\Salvatore Ingrassia\thanks{Department of Economics and Business, University of Catania, Catania, Italy. E-mail: s.ingrassia@unict.it}, Ryan P.\ Browne\thanks{Department of Statistics and Actuarial Science, University of Waterloo, Waterloo, Ontario, Canada. E-mail: rpbrowne@uwaterloo.ca} }
%\email{ \{rbrowne,paul.mcnicholas\}@uoguelph.ca}
%\address{ 
%\author{Ryan P.\ Browne and Paul D.\ McNicholas \\
%\date{Department of Mathematics \& Statistics, University of Guelph}
\date{}
\maketitle
\begin{abstract}
A family of parsimonious Gaussian cluster-weighted models is presented. This family concerns a multivariate extension to cluster-weighted modelling that can account for correlations between multivariate responses. Parsimony is attained by constraining parts of an eigen-decomposition imposed on the component covariance matrices. A sufficient condition for identifiability is provided and an expectation-maximization algorithm is presented for parameter estimation. Model performance is investigated on both synthetic and classical real data sets and compared with some popular approaches. Finally, accounting for linear dependencies in the presence of a linear regression structure is shown to offer better performance, vis-\`{a}-vis clustering, over existing methodologies.
\end{abstract}
%\newpage

\section{Introduction}\label{sec:introduction}
Mixture models have seen increasing use over the decade or two, with many important applications in clustering and classification (examples include: \citealp{hennig13}, \citealp{lee13}, \citealp{andrews14}, \citealp{biernacki2014}, \citealp{browne14b}, \citealp{bouveyron2014}, \citealp{lin14}, \citealp{franczak15}, \citealp{wei15}, \citealp{anderlucci15}, \citealp{ohagan16}, and \citealp{mcnicholas16}). Arguably, the most famous model-based clustering methodology is the Gaussian parsimonious clustering models (GPCM) family \citep{Cele:Gova:Gaus:1995}, which is supported by the {\tt mclust} \citep{fraley2002,mclust}, {\tt mixture} \citep{mixture1.1}, and {\tt Rmixmod} \citep{rmixmod} packages for {\sf R} \citep{R2015}. However, such models do not typically account for dependencies via covariates. When there is a clear regression relationship between some variables, important insight can be gained by accounting for functional dependencies between those variables. For such data, traditional model-based clustering methods that fail to incorporate such a relationship may not perform as well. 

Some popular mixture-based methodologies that deal with regression data are finite mixtures of regression \citep[FMR;][]{desarbo1988} and finite mixtures of regression with concomitant variables \citep[FMRC;][]{wedel2002}, supported by the {\tt flexmix} \citep{leisch2004, flexmix} package for {\sf R}. FMR only model the distribution of the response given the covariates, whereas FMRC also model the mixing weights of the components as a multinomial logistic model of some concomitant variables (which are often the covariate variables). However, these methodologies do not explicitly use the distribution of the covariates for clustering, i.e., the assignment of data points to clusters does not directly utilize any information from the distribution of the covariates. Recently, finite mixtures of seemingly unrelated linear regressions have also been proposed \citep{galimberti2015}.

A flexible framework for density estimation and clustering of data with local functional dependencies is represented by the cluster-weighted model \citep[CWM;][]{Gers:Nonl:1997}, also called the saturated mixture regression model by \citet{wedel2002}. CWMs were first investigated in a general statistical mixture framework by \citet{ingrassia2012}. The same paper presented theoretical and numerical properties, and discussed the performance of the model under both Gaussian and $t$-distributional assumptions \citep[see also][]{Ingr:Mino:Punz:Mode:2014}. As opposed to FMR and FMRC, CWMs allow for \textit{assignment dependence} \citep[cf.][]{Henn:Iden:2000}, i.e., the assignment of an observation to a cluster is also dependent on the distribution of the covariates. In such models, the component covariate distributions can be distinct; for a discussion about the difference between FMR, FMRC, and CWMs from a geometrical point of view, see \citet{Ingr:Punz:Deci:2016}. Some extensions of this methodology have dealt with non-linear local relationships \citep{Punz:Flex:2014}, high dimensional covariates \citep{Sube:Punz:Ingr:McNi:Clus:2013,Sube:Punz:Ingr:McNi:Clus:2015}, and various response types \citep{Ingr:Punz:Vitt:TheG:2014,Punz:Ingr:OnTh:2013,Punz:Ingr:Clus:2015}.
%, and data contaminated by atypical observations \citep{Punz:McNi:Robu:2015}.

In this paper, a family of parsimonious Gaussian CWMs is presented. 
It concerns multivariate response, while CWMs have so far only dealt with univariate response. Note that parsimony is vital in real data applications and this is the first time that CWMs are being used with eigen-decomposed covariance structures. Regarding multivariate response in the FMR and FMRC framework, using multivariate response is possible in \texttt{flexmix} but the package currently does not account for correlated response variables, i.e., these models assume independence between the response variables. FMR models that deal with correlated response variables have been recently proposed \citep{soffritti2011, galimberti2013}, but these models do not decompose the covariance structure nor do they use information from the distribution of the covariates. Furthermore, these papers do not investigate FMRC models that deal with correlated responses. Families of eigen-decomposed parsimonious FMR and FMRC models that can account for correlated response variables have been recently proposed \citep[eFMR and eFMRC;][]{dang2015}; however, these models do not take into account the distribution of the covariates. 

For the proposed multivariate response CWM, parsimonious models are developed by constraining parts of an eigen-decomposition imposed on the component covariance matrices of both the responses and the covariates. This family of parsimonious models is referred to as the eigen-decomposed multivariate response CWM (eMCWM).
The Bayesian information criterion \citep[BIC;][]{schwarz1978} and the integrated completed likelihood \citep[ICL;][]{biernacki2000} are considered for model selection on this family. Comparisons to FMR, FMRC, eFMR, eFMRC, and GPCMs are made. Note that, hereafter, FMR and FMRC models as implemented in {\tt flexmix} will simply be referred to as the FMR and FMRC models, respectively. 

The remainder of the paper is organized as follows. 
In Section~\ref{sec:methodology}, basic ideas on CWMs are summarized. 
In Section~\ref{sec:parsimonious_models}, we recall an eigen-decomposition of a covariance matrix.
Identifiability is treated in Section~\ref{sec:Identifiability}. 
An expectation-maximization (EM) algorithm for maximum likelihood parameter estimation is presented in Section~\ref{sec:Inference}. Moreover, issues of model selection, algorithm initialization, convergence criterion, and performance assessment are also discussed. 
In Section~\ref{sec:numerical_studies}, results of numerical studies based on both real and simulated data are presented. 
Finally, in Section~\ref{sec:discussion}, some conclusions and ideas for future research are discussed.

\section{Cluster-weighted models}
\label{sec:methodology}

Multivariate correlated responses can be conveniently accounted for in a CWM framework. Let $\bX$ and $\bY$ be random vectors defined on $\Omega$ with joint probability distribution $p(\bx,\by)$. Here, the response vector $\bY$ has values in $\mathbb{R}^d$ and the vector of covariates $\bX$ has values in $\mathbb{R}^p$. Let $\Omega$ be partitioned into $G$ disjoint groups, such that $\Omega= \Omega_1 \cup \cdots \cup \Omega_G$. Then, in a CWM framework, the joint probability $p(\bx,\by)$ can be decomposed as
\begin{equation}
p(\bx,\by)=\sum^{G}_{g=1} p(\by|\bx,\Omega_g) p(\bx|\Omega_g) \pi_g, 
\label{eq:MCWM}
\end{equation}
where $p(\by|\bx,\Omega_g)$ is the conditional density of the multivariate response $\bY$ given the covariates $\bX$ and $\Omega_g$, $p(\bX|\Omega_g)$ is the probability density of $\bX$ given $\Omega_g$, and $\pi_g=p(\Omega_g)$ are the mixing weights, where $\pi_{g}>0$ and $\sum^G_{g=1}\pi_{g}=1$, $g = 1,\ldots,G$. 
Here, $\bX|\Omega_g$ is assumed to be normally distributed with mean $\bmu_{\bX g}$ and covariance matrix $\bSigma_{\bX g}$, and $\bY|(\bX=\bx,\Omega_g)$ is assumed to be normally distributed with conditional mean $\bmu_{\bY}\left(\bx|\bbeta_g\right)$, given by some linear transformation of $\bX$, and covariance matrix $\bSigma_{\bY g}$, $g = 1,\ldots,G$.
Here, $\bmu_{\bY}\left(\bx|\bbeta_g\right) = \bbeta'_g \bx^*$ is used where $\bbeta_g \in \mR^{(1+p) \times d}$ and $\bx^*=\left(1,\bx'\right)'$. Hence, $\bY$ ($\bX$) is a matrix of $N$ observations on $d$ ($p+1$) response variables while $\bbeta$ for a specific component $g$ is a $p+1$ by $d$ matrix of regression coefficients with one column for each response variable.
%Afterwards, we set $\bbeta_g = \{\bbeta_{0g}, \bbeta'_{1g}\}$. 
Then, model~\eqref{eq:MCWM} can be rewritten as
\begin{equation}
p\left(\bx,\by|\bvartheta\right)=\sum^{G}_{g=1} 
\phi_d\left(\by|\bx,\bmu_{\bY}(\bx|\bbeta_g),\bSigma_{\bY g}\right) 
\phi_p\left(\bx|\bmu_{\bX g},\bSigma_{\bX g}\right)  
\pi_g, 
\label{mGCWM}
\end{equation}
where $\phi_d$ ($\phi_p$) represents the density of a $d$-variate ($p$-variate) Gaussian random vector and $\bvartheta$ denotes the set of all parameters.

\section{Parsimonious Models}
\label{sec:parsimonious_models}

For a single $q\times q$ covariance matrix, the number of free parameters increases quadratically with the dimensionality $q$. In model-based clustering, parsimony is usually necessary for real applications. Parsimony can be introduced by constraining parts of a particular decomposition of a covariance matrix \citep{Cele:Gova:Gaus:1995,mcnicholas2010,Punz:Brow:McNi:Hypo:2015}. An eigen-decomposition of such a matrix \citep[cf.][]{Cele:Gova:Gaus:1995} yields
\begin{equation*}
\bSigma_g=\lambda_g\boldsymbol{\Gamma}_g\boldsymbol{\Delta}_g\boldsymbol{\Gamma}_g',
\end{equation*}
for $g=1,\ldots,G$, where $\lambda_g=\left|\boldsymbol{\Sigma}_g\right|^{1/q}$ is a constant, $\boldsymbol{\Delta}_g$
is a diagonal matrix with entries (sorted in decreasing order) proportional to the eigenvalues of $\bSigma_g$ with the constraint $|\boldsymbol{\Delta}_g|=1$, and $\boldsymbol{\Gamma}_g$ is a $q\times q$ orthogonal matrix of the eigenvectors (ordered according to the eigenvalues) of~$\bSigma_{g}$.
Geometrically, $\lambda_g$ determines the volume, $\boldsymbol{\Delta}_g$ the shape, and $\boldsymbol{\Gamma}_g$ the orientation of the $g$th component. By constraining $\lambda_g$, $\boldsymbol{\Gamma}_g$, and $\boldsymbol{\Delta}_g$ to be equal or variable across groups, a family of fourteen models (\tablename~\ref{tab:models}) is obtained. This family can be further split into three subfamilies. Here, the EII and VII models belong to the \textit{spherical} family, the models with an axis-aligned orientation belong to the \textit{diagonal} family, while the rest of the models belong to the \textit{general} family.
\begin{table}[htb]
\caption{Geometric interpretation and the number of free parameters in the eigen-decomposed covariance structures.} \label{tab:models}
%\smallskip
%\resizebox*{1\textwidth}{!}{
\begin{tabular*}{1.0\textwidth}{@{\extracolsep{\fill}}lllllr}
\hline
Model & Volume & Shape & Orientation & $\boldsymbol{\Sigma}_g$ & Free Cov.\ Parameters \\
\hline
EII & Equal    & Spherical & -            & $\lambda \boldsymbol{I}$           & 1\\
VII & Variable & Spherical & -            & $\lambda_g \boldsymbol{I}$         & $G$\\[2mm]
EEI & Equal    & Equal     & Axis-Aligned & $\lambda \boldsymbol{\Delta}$      & $q$\\
VEI & Variable & Equal     & Axis-Aligned & $\lambda_g \boldsymbol{\Delta}$    & $G+q-1$\\
EVI & Equal    & Variable  & Axis-Aligned & $\lambda \boldsymbol{\Delta}_g$    & $Gq - (G - 1)$\\
VVI & Variable & Variable  & Axis-Aligned & $\lambda_g \boldsymbol{\Delta}_g$  & $Gq$\\[2mm]
EEE & Equal    & Equal     & Equal        & $\lambda\boldsymbol{\Gamma}\boldsymbol{\Delta}\boldsymbol{\Gamma}'$  & $q\left(q+1\right)/2$\\
VEE & Variable & Equal     & Equal        & $\lambda_g\boldsymbol{\Gamma}\boldsymbol{\Delta}\boldsymbol{\Gamma}'$  & $q(q+1)/2 + (G-1)$ \\
EVE & Equal    & Variable  & Equal        & $\lambda\boldsymbol{\Gamma}\boldsymbol{\Delta}_g\boldsymbol{\Gamma}'$  & $q(q+1)/2 + (G-1)(q-1)$ \\
EEV & Equal    & Equal     & Variable     & $\lambda\boldsymbol{\Gamma}_g\boldsymbol{\Delta}\boldsymbol{\Gamma}_g'$  & $Gq(q+1)/2 - (G-1)q$ \\
VVE & Variable & Variable  & Equal        & $\lambda_g\boldsymbol{\Gamma}\boldsymbol{\Delta}_g\boldsymbol{\Gamma}'$  & $q(q+1)/2 + (G-1)q$ \\
VEV & Variable & Equal     & Variable   & $\lambda_g\boldsymbol{\Gamma}_g\boldsymbol{\Delta}\boldsymbol{\Gamma}_g'$  & $Gq(q+1)/2 - (G-1)(q-1)$ \\
EVV & Equal    & Variable  & Variable   & $\lambda\boldsymbol{\Gamma}_g\boldsymbol{\Delta}_g\boldsymbol{\Gamma}_g'$  & $Gq(q+1)/2 - (G-1)$ \\
VVV & Variable & Variable  & Variable   & $\lambda_g \boldsymbol{\Gamma}_g \boldsymbol{\Delta}_g \boldsymbol{\Gamma}_g'$  & $Gq\left(q+1\right)/2$ \\
\hline
\end{tabular*}
%}
\end{table}

Here, the covariance matrices $\bSigma_{\bX g}$ and $\bSigma_{\bY g}$ in \eqref{mGCWM} are decomposed. Constraining $\lambda_g$, $\boldsymbol{\Gamma}_g$, and $\boldsymbol{\Delta}_g$ on these decompositions in Equation~\eqref{mGCWM} leads to 14 different covariance structures for both $\bX$ and $\bY$, resulting in a total of $14\times 14=196$ models. This is the eMCWM family.  
Note that for the purposes of notation, an eMCWM with a VEV covariance structure for $\bY|\bX=\bx$ and an EII covariance structure for $\bX$ will be denoted as a VEV-EII model.

\section{Identifiability}
\label{sec:Identifiability}

Identifiability is important for parameter inference and for the usual asymptotic theory to hold for maximum likelihood estimation of the model parameters (cf.\ Section~\ref{sec:Inference}).
Proof of the identifiability of univariate and multivariate finite Gaussian mixture distributions is provided by \citet{Teic:Iden:1963} and \citet{Yako:Spra:Onth:1968}, respectively, while general conditions for identifiability of mixtures of linear models can be found in \citet{Henn:Iden:2000}. 
Proof of the identifiability of the generalized linear Gaussian CWM has recently been provided by \citet{Ingr:Punz:Vitt:TheG:2014}. 
Here, identification conditions are provided for the multivariate response (Gaussian) CWM defined in \eqref{mGCWM}.

Generally speaking, identifiability for mixture models can be defined as follows. 
Consider a parametric class of density (probability) functions $\cF = \left\{ f(\bz| \bpsi) \, : \, \bz \in \cZ \, , \, \bpsi \in \bPsi \right\}$. Then, the class of finite mixtures of functions in $\cF$ is
\begin{align*} 
\cH & = \biggl\{ h(\bz| \bvartheta) \, : \, h(\bz|\bvartheta) = \sum_{g=1}^G f(\bz| \bpsi_g) \pi_g, \, \mbox{ with } \,  \pi_g >0 \, \mbox{ and } \, \sum_{g=1}^G 
\pi_g =1,  \nonumber \\
& \qquad   f(\cdot| \bpsi_g) \in \cF, g=1, \ldots, G, \bpsi_g \neq \bpsi_j \mbox{ for } g\neq j,  \, G \in \mN, \, \bz \in \cZ\, , \, \bvartheta\in\bTheta  \biggr\}.  
\end{align*} 
This class is identifiable if for any two members 
\[ 
h(\bz|\bvartheta)=\sum_{g=1}^{G} f(\bz| \bpsi_g) \pi_g \quad \mbox{and} \quad 
h(\bz|\widetilde{\bvartheta}) = \sum_{s=1}^{\widetilde{G}} f(\bz| \widetilde{\bpsi}_s) \widetilde{\pi}_s 
\]
of $\cH$, $h(\bz| \bvartheta) = h(\bz| \widetilde{\bvartheta})$ implies that $G=\widetilde{G}$, and for each $g \in \{1, \ldots, G\}$ there exists $s\in\{1,\ldots,G\}$ such that $\pi_g=\widetilde{\pi}_s$ and $\bpsi_g=\widetilde{\bpsi}_s$. 

In Theorem~\ref{thm:identif}, a sufficient identification condition is provided for the most general eMCWM (i.e., the VVV-VVV model).
In particular, a sufficient condition for the identifiability of the class $\cC$ is established where
\begin{align} 
\cC & = \biggl\{ p(\bx, \by| 	\bvartheta) \, : \, p(\bx, \by| \bvartheta) = \sum^{G}_{g=1} 
\phi_d\left(\by|\bx, \bmu_{\bY}\left(\bx|\bbeta_g\right),\bSigma_{\bY g}\right) \phi_p(\bx|\bmu_{\bX g},\bSigma_{\bX g}) \pi_g, 
\biggr. \nonumber \\& 
\qquad \biggl.  
\text{ with }  \pi_g >0, \sum_{g=1}^G \pi_g =1, \, (\bbeta_g,\bSigma_{\bY g})  \neq (\bbeta_j,\bSigma_{\bY j})  \mbox{ for } g \neq j, 
\, \left(\bx, \by\right) \in \mR^{p+d}\, , \, 
 \biggr. \nonumber \\& \qquad \biggl.  
 \bvartheta = \{\bbeta_g, \bSigma_{\bY g}, \bmu_{\bX g}, \bSigma_{\bX g}, \pi_g; g=1, \ldots, G\} \in\bTheta,  G \in \mN \, \biggr\}.  
\label{eq:class_gencwm}
\end{align} 
In the following theorem, a sufficient condition for  $\cC$ to be identifiable in $\cX \times \mR^d$ is provided, where  $\cX \subseteq \mR^p$ is a set having probability one according to the multivariate Gaussian density $\phi_p$.

In other words, this proves that the class $\cC$ is identifiable for almost all $\bx \in \mR^p$ and for all $\by \in \mR^d$.
\begin{Thm}\label{thm:identif}{\rm Let $\cC$ be the class  defined in \eqref{eq:class_gencwm} and assume that there exists a set $\cX \subseteq \mR^p$ having probability equal to one according to the $p$-variate Gaussian distribution such that  the mixture of regression models
\begin{equation}
\sum_{g=1}^G \phi_d\left(\by|\bx, \bmu_{\bY}\left(\bx|\bbeta_g\right),\bSigma_{\bY g}\right) \alpha_g\left(\bx\right),  \quad \by \in \mR^d,  
\label{eq:mixtglm_alpha}
\end{equation}
is identifiable for  each fixed  $\bx \in \cX$, where $\alpha_1(\bx), \ldots, \alpha_G(\bx)$ are positive weights summing to one for each $\bx \in \cX$. 
Then, the class $\cC$ is identifiable in $\cX \times \mR^d$.
}\end{Thm}
\begin{proof}
The proof is given in \appendixname~\ref{subsec:identifiability}.
\end{proof}

\section{Inference}
\label{sec:Inference}

\subsection{Parameter Estimation for eMCWM}
\label{sec:parameter_estimation}

Parameter estimation, via the EM algorithm of \citet{Demp:Lair:Rubi:Maxi:1977}, is described here for the unconstrained VVV-VVV model from the eMCWM family. Details on alternative algorithms to find maximum likelihood estimates of a mixture distribution can be found in \citet{bohning1995, bohning2003}.
Let $\mathcal{S}=\left\{(\bx_1,\by_1),\ldots,(\bx_N,\by_N)\right\}$ be a sample of $N$ independent observations from model~\eqref{mGCWM}. 
Then, the incomplete-data likelihood function is
\begin{equation*}
L\left(\bvartheta|\mathcal{S}\right) 
= 
\prod_{i=1}^N p(\bx_i,\by_i|\bvartheta)  
= 
\prod_{i=1}^N \left[\sum_{g=1}^G 
\phi_d(\by_i|\bmu_{\bY}(\bx_i|\bbeta_g),\bSigma_{\bY g}) 
\phi_p(\bx_i|\bmu_{\bX g},\bSigma_{\bX g})  
\pi_g 
\right].
\label{likCWMg} 
\end{equation*}
%where, for notational convenience, $\bmu_{\bX g}$ has been substituted by $\bbeta_g$.
% $\bbeta_g,\bSigma_{\bY g}=(\bbeta_g, \bSigma_{\bY g})$ refers to the parameters of the distribution of the response given the covariates. 
%Here, the covariates $\bx$ are supplemented by a vector of ones such that $\bbeta_g$ is the matrix of regression intercepts and coefficients. 
%Similarly, $\bpsi_g=(, )$ contains the mean and covariance for the covariates. 
Note that $\mathcal{S}$ is considered incomplete in the context of the EM algorithm. The complete-data are $\mathcal{S}_{\text{c}} = \left\{(\bx_1, \by_1, \bz_1),\ldots,(\bx_N, \by_N, \bz_N)\right\}$, where the missing variable $\bz_{i}$ is the component label vector such that $z_{ig}$ = 1 if $(\bx_i,\by_i)$ belongs to the $g$th component and $z_{ig} = 0$ otherwise. 
Now, the corresponding complete-data likelihood is
\begin{equation*}
L_c\left(\bvartheta|\mathcal{S}_{\text{c}}\right) = 
\prod_{i=1}^N \prod_{g=1}^G [
\phi_d(\by_{i}|\bmu_{\bY}(\bx_i|\bbeta_g),\bSigma_{\bY g})
\phi_p(\bx_i|\bmu_{\bX g},\bSigma_{\bX g})  \pi_g]^{z_{ig}}.
\end{equation*}
The complete-data log-likelihood function can be decomposed as
\begin{displaymath}
l_c\left(\bvartheta|\mathcal{S}_{\text{c}}\right) = \sum_{i=1}^N \sum_{g=1}^G 
z_{ig} 
\left[ 
\log \phi_d(\by_{i}|\bmu_{\bY}(\bx_i|\bbeta_g),\bSigma_{\bY g}) + 
\log \phi_p(\bx_i|\bmu_{\bX g},\bSigma_{\bX g}) + 
\log \pi_g\right] 
.
\end{displaymath}
The E-step involves calculating the expected complete-data log-likelihood
\begin{align*}
Q\left(\bvartheta|\bvartheta^{(k)}\right) 
& = \EE_{\bvartheta^{(k)}}\left\{l_c\left(\bvartheta|\mathcal{S_{\text{c}}}\right)\right\} \\
%& = \sum_{i=1}^N \sum_{g=1}^G \EE_{\bvartheta^{(k)}} \{Z_{ig}|\bx_i, \by_i\} \left[ Q_{1}\left(\bbeta_g,\bSigma_{\bY g}|\bvartheta^{(k)}\right) + Q_{2}\left(\bmu_{\bX g},\bSigma_{\bX g}|\bvartheta^{(k)}\right) + \log \pi_g^{(k)} \right] \notag \\ 
& = \sum_{i=1}^N \sum_{g=1}^G  \hat{z}_{ig}^{(k)} \left[ Q_{1}\left(\bbeta_g,\bSigma_{\bY g}|\bvartheta^{(k)}\right) + Q_{2}\left(\bmu_{\bX g},\bSigma_{\bX g}|\bvartheta^{(k)}\right) + \log \pi_g^{(k)}\right],
\end{align*}
where
\begin{equation*} 
\hat{z}_{ig}^{(k)} =  \EE_{\bvartheta^{(k)}} \{Z_{ig}|\bx_i, \by_i\} = 
\frac{
\phi_d(\by_i|\bmu_{\bY}(\bx_i|\bbeta_g^{(k)}),\bSigma_{\bY g}^{(k)}) 
\phi_p(\bx_i|\bmu_{\bX g}^{(k)},\bSigma_{\bX g}^{(k)})
\pi_g^{(k)} 
} 
{
\displaystyle\sum_{j=1}^G 
\phi_d(\by_i|\bmu_{\bY}(\bx_i|\bbeta_j^{(k)}),\bSigma_{\bY j}^{(k)}) 
\phi_p(\bx_i|\bmu_{\bX j}^{(k)},\bSigma_{\bX j}^{(k)})
\pi_j^{(k)}
} 
\end{equation*}
provides the current value of $z_{ig}$ on the $k^{\text{th}}$-iteration and
\begin{align*} 
Q_{1}\left(\bbeta_g,\bSigma_{\bY g}|\bvartheta^{(k)}\right) & =  \frac{1}{2} \left[- d \log \left(2\pi\right) - \log |\bSigma_{\bY g}^{(k)}| - \left(\by_i-\bbeta^{'(k)}_g \bx_i^*\right)' \bSigma_{\bY g}^{(k)(-1)} (\by_i-\bbeta^{'(k)}_g \bx_i^*)\right], \\
Q_{2}\left(\bmu_{\bX g},\bSigma_{\bX g}|\bvartheta^{(k)}\right) & =  \frac{1}{2} \left[- p \log \left(2\pi\right) - \log |\bSigma_{\bX g}^{(k)}| - \left(\bx_i-\bmu_{\bX g}^{(k)}\right)' \bSigma_{\bX g}^{(k)(-1)} \left(\bx_i-\bmu_{\bX g}^{(k)}\right)\right].
\end{align*}
The M-step on the $(k+1)^\text{th}$ iteration of the EM algorithm involves the maximization of the conditional expectation of the complete-data log-likelihood with respect to $\bvartheta$.
The update for $\pi_{g}^{(k+1)}$ is
\begin{equation}
\hat{\pi}_{g}^{(k+1)} = \frac{1}{N} \sum_{i=1}^N \hat{z}_{ig}^{(k)}. \label{pi_g} 
\end{equation}
The updates for $\bmu_{\bX g}^{(k+1)}$ and $\bSigma_{\bX g}^{(k+1)}$, $g=1, \ldots, G$, are
\begin{align}
\hat{\bmu}_{\bX g}^{(k+1)} &= \frac{\sum_{i=1}^N \hat{z}_{ig}^{(k)}  \bx_i}{\sum_{i=1}^N \hat{z}_{ig}^{(k)} }, \label{bmu_g} \\
\hat{\bSigma}_{\bX g}^{(k+1)} &= \frac{\sum_{i=1}^N \hat{z}_{ig}^{(k)}  \left(\bx_i - \hat{\bmu}_{\bX g}^{(k+1)}\right)\left(\bx_i - \hat{\bmu}_{\bX g}^{(k+1)}\right)'}{\sum_{i=1}^N \hat{z}_{ig}^{(k)} }.  
\label{bSigma_xg}
\end{align}
These closed form updates can also be found in \citet{mclachlan2000}.
The updates for $\bbeta_{g}^{(k+1)}$ and $\bSigma_{\bY g}^{(k+1)}$ (see \appendixname~\ref{mstepderivation} for details), $g=1, \ldots, G$, are
\begin{equation}
\hat{\bbeta}_{g}^{(k+1)'} = \left( \sum_{i=1}^N \hat{z}_{ig}^{(k)}  \by_i \bx_i^{*'}  \right)  
\left({\sum_{i=1}^N \hat{z}_{ig}^{(k)}  \bx_i^* \bx_i^{*'}}\right)^{-1}
\label{bB_g}
\end{equation}
when ${\sum_{i=1}^N \hat{z}_{ig}^{(k)}  \bx_i^* \bx_i^{*'}}$ is non-singular, and
\begin{equation}
\hat{\bSigma}_{\bY g}^{(k+1)} = \frac{\sum_{i=1}^N \hat{z}_{ig}^{(k)}  \left(\by_i-\hat{\bbeta}^{(k+1)'}_g \bx_i^*\right)\left(\by_i-\hat{\bbeta}^{(k+1)'}_g \bx_i^*\right)'}{\sum_{i=1}^N \hat{z}_{ig}^{(k)} },
\label{bSigma_yg}
\end{equation}
respectively. 
Equations \eqref{pi_g} through \eqref{bSigma_yg} are the parameter updates for the VVV-VVV model of the eMCWM family. 
For the other models of this family, the M-step updates vary only with respect to the component covariance matrices $\bSigma_{\bX g}$ and $\bSigma_{\bY g}$; these updates are similar to those of the GPCM family of \citet{Cele:Gova:Gaus:1995}.

\subsection{Model Selection}
\label{sec:Model Selection}

For choosing the ``best'' fitted model among a family of models, a likelihood-based model selection criterion is conventionally used, and the BIC is the most popular for Gaussian mixture models. Even though mixture models generally do not satisfy the regularity conditions for the asymptotic approximation used in the development of the BIC \citep{keribin1998, keribin2000}, it has performed well in practice and has been used extensively since the work of \cite{dasgupta1998} and \cite{fraley2002}. 
The BIC can be calculated as 
$$
\text{BIC}=2l(\hat{\bvartheta})-m\log{N},
$$
where $l(\hat{\bvartheta})$ is the incomplete-data log-likelihood at the maximum likelihood estimates and $m$ is the number of free parameters.
%, and $N$ is the sample size. 
The ICL is another commonly used information criterion that additionally makes use of the estimated mean entropy, i.e., it takes into account the uncertainty of the classification of an observation to a component and can be computed as 
$$
\text{ICL}\approx\mbox{BIC}+\sum_{i=1}^{N}\sum_{g=1}^{G}\mbox{MAP}(\hat{z}_{ig})\log{\hat{z}_{ig}}.
$$
Here, $\mbox{MAP}(\hat{z}_{ig})$ is the maximum \emph{a posteriori} probability and equals 1 if $\mbox{max}_{h}(\hat{z}_{ih})$, $h=1,\ldots,G$, occurs at component $h=g$, and 0 otherwise. 

\subsection{Initialization}
\label{sec:initialization}

The EM algorithm is noted to be heavily dependent on starting values \citep{baudry2015}. Singularities and convergence to local maxima are well documented, and Gaussian mixture models are known to have unbounded likelihood surfaces \citep{titterington1985}. 
Constraining eigenvalues can alleviate some of these issues \citep{ingrassia2007, browne2013constrained}, as can employing deterministic annealing \citep{zhou2010}.
Initializing the EM algorithm multiple times using $k$-means \citep{macqueen1967, hartigan1979} or random initializations and choosing the initial values of $z_{ig}$ from the run picked using the highest log-likelihood value can also help. 

Here, the EM algorithm is initialized as follows. 
The EEE-EEE model is run 10 times for each $G$: 9 times using a random initialization for the $z_{ig}$, and once with a $k$-means initialization. Note that, for the $k$-means initialization, the initial $\tau_{ig}$ are selected from the best $k$-means clustering results from ten random starting values for the $k$-means algorithm as implemented in {\sf R}. From these models, the model with the highest log-likelihood value is chosen; then, the associated $\mbox{MAP}(\hat{z}_{ig})$ is used to initialize the families of models.  
In our simulations, this procedure performed well.

\subsection{Convergence Criterion}

A common criterion to stop the EM algorithm is when the difference between the log-likelihood values on consecutive iterations is less than some $\epsilon$. Here, the Aitken stopping criterion \citep{Aitk:OnBe:1926} is used to determine convergence. Use of the aforementioned lack of progress criterion can result in the EM converging earlier than with the Aitken stopping criterion, resulting in estimates that might not be as close to the maximum likelihood estimates \citep{mcnicholas2010}. At each iteration of the EM algorithm, the Aitken acceleration procedure is used to compute an estimated asymptotic value. Based on this, a decision can be made regarding whether or not the algorithm has reached convergence, i.e., whether or not the log-likelihood is sufficiently close to its estimated asymptotic value. See \citet{bohning1994}, \citet{lindsay1995}, and \citet{mcnicholas2010} for details.

\subsection{Performance Assessment}

The adjusted Rand index \citep[ARI;][]{hubert1985} can be used to judge performance of a model relative to the true classification (when known). The predicted group memberships at the maximum likelihood estimates of the model parameters are given by %the maximum \emph{a posteriori} probability 
$\mbox{MAP}(\hat{z}_{ig})$.%, which equals 1 if $\mbox{max}_{h}(\hat{z}_{ih})$, $h=1,\ldots,G$, occurs at component $g$, and 0 otherwise. 
The Rand index (RI) can be used to compare these partitions \citep{rand1971}. These predicted classifications from the best fitted model can be cross-tabulated against the true (known) group memberships. Now, defining agreements based on pairwise comparisons as the observations that should be in the same group and are, plus those that should not be in the same group and are not, the Rand index is calculated as the number of agreements over the total number of pairs. The RI takes a value between 0 and 1 (the latter indicative of perfect agreement). However, the RI has a positive expected value under random assignment, which leads to difficulty in interpreting small RI values.

The ARI calculates the agreement between true and predicted classification by correcting the RI to account for chance. Hence, an ARI of 1 corresponds to perfect clustering whereas an ARI of 0 implies that the results are no better than would be obtained by chance. Furthermore, the ARI can be negative and such values are indicative of classification that is worse than would be expected by random assignment. The reader is advised to consult \cite{steinley2004} for discussion and further details about the ARI. 

\section{Experiments and Illustrations}
\label{sec:numerical_studies}

The eMCWM family is implemented in {\sf R}. Simulations are presented to illustrate parameter recovery for a host of different models. Clustering performance of the eMCWM family is compared to that of FMR, FMRC, eFMR, eFMRC, and GPCMs on benchmark data. Although the objective was to develop a mixture model that can incorporate linear dependencies on covariates, the models' performance is also compared to the GPCM family as implemented in the {\tt mixture} package. Note that GPCMs are constructed on different assumptions and are not meant to account for regression structures. However, GPCMs remain the most commonly used models in the literature for model-based clustering and, hence, a comparison to these models is relevant. As mentioned earlier, these models are supported by the \texttt{mclust}, \texttt{mixture}, and \texttt{Rmixmod} packages. In particular, \texttt{mixture} makes use of a majorization-minimization algorithm for the EVE and VVE models  \citep{browne2014mm}, which works better in higher dimensions than the Flury procedure \citep{Flur:Gaut:anal:1986} used in \texttt{Rmixmod}.
Note that,	 for the M-step for the different covariance structures in \tablename~\ref{tab:models}, the \texttt{mixture} package \citep{mixture1.1} is used. 
The eMCWM family is also compared to the eFMR and eFMRC families of \citet{dang2015}. 
The \texttt{flexmix} FMR and FMRC algorithms allow for specification of a user-defined initialization matrix; thus, all
algorithms run on a specific data sample are initialized with the same set of $\mbox{MAP}(\hat{\tau}_{ig})$ values to facilitate comparison of the performance of the algorithms from the same starting values (Section~\ref{sec:initialization}).
The performance of the eMCWM family is illustrated on artificial and real data sets in Sections \ref{sec:simulations} and \ref{sec:realdata}, respectively.

\subsection{Analyses on Simulated Data} \label{sec:simulations}
\subsubsection{Simulation 1} \label{sec:VEEVII}

Here, data of size $N=250$ and with $p=d=2$ are generated from a two-component VEE-VII model (\figurename~\ref{simscat}). 
The data set will be referred to as Simulation 1 hereafter. A two-component mixture is simulated with the sample sizes for each group sampled from a binomial distribution with success probability 0.35.
The responses are generated using a VEE covariance structure and the covariates using a VII covariance structure. 
Covariates are generated from a bivariate Gaussian distribution with mean $\bmu_{\bX1}=(3, 2.5)'$ and $\bmu_{\bX2}=(1.1, -4)'$ for components 1 and 2, respectively. The covariance matrices of the covariates for the two groups are $\bigl(\begin{smallmatrix}
1&0\\ 
0&1
\end{smallmatrix} \bigr)$ and 
$\bigl(\begin{smallmatrix}
0.5&0\\ 
0&0.5
\end{smallmatrix} \bigr)$, 
respectively. Under the VII decomposition, this corresponds to $\lambda_{\bX1}=1$ and $\lambda_{\bX2}=0.5$ for component~1 and component~2, respectively. 
The regression coefficient matrices used for the two groups are $\bigl(\begin{smallmatrix}
2 & -0.5 & -1\\ 
-2 & 1.5 & 2
\end{smallmatrix} \bigr)'$ and 
$\bigl(\begin{smallmatrix}
0 & 2.2 & -1\\ 
1 & 2 & 1.5
\end{smallmatrix} \bigr)'$, 
respectively. Lastly, the error matrices for the two groups using a VEE covariance structure are $\bigl(\begin{smallmatrix}
0.92 & 0.56\\ 
0.56 & 1.40
\end{smallmatrix} \bigr)$ and 
$\bigl(\begin{smallmatrix}
1.725 & 1.050\\ 
1.050 & 2.625
\end{smallmatrix} \bigr)$, 
respectively. 
This corresponds to $\lambda_{\bY 1}=0.8$, $\lambda_{\bY 2}=1.5$, and $\boldsymbol{\Gamma}_{\bY 1}\boldsymbol{\Delta}_{\bY 1}\boldsymbol{\Gamma}_{\bY 1}'=\boldsymbol{\Gamma}_{\bY 2}\boldsymbol{\Delta}_{\bY 2}\boldsymbol{\Gamma}_{\bY 2}'=\bigl(\begin{smallmatrix}
1.15 & 0.70\\ 
0.70 & 1.75
\end{smallmatrix} \bigr)$.

\begin{figure}[htb]
\begin{center}
\resizebox{0.7\textwidth}{!}{
\includegraphics[angle=-90] %,origin=c
%[trim=0cm 3cm 0cm -2cm,clip,scale=0.5,angle=270]
{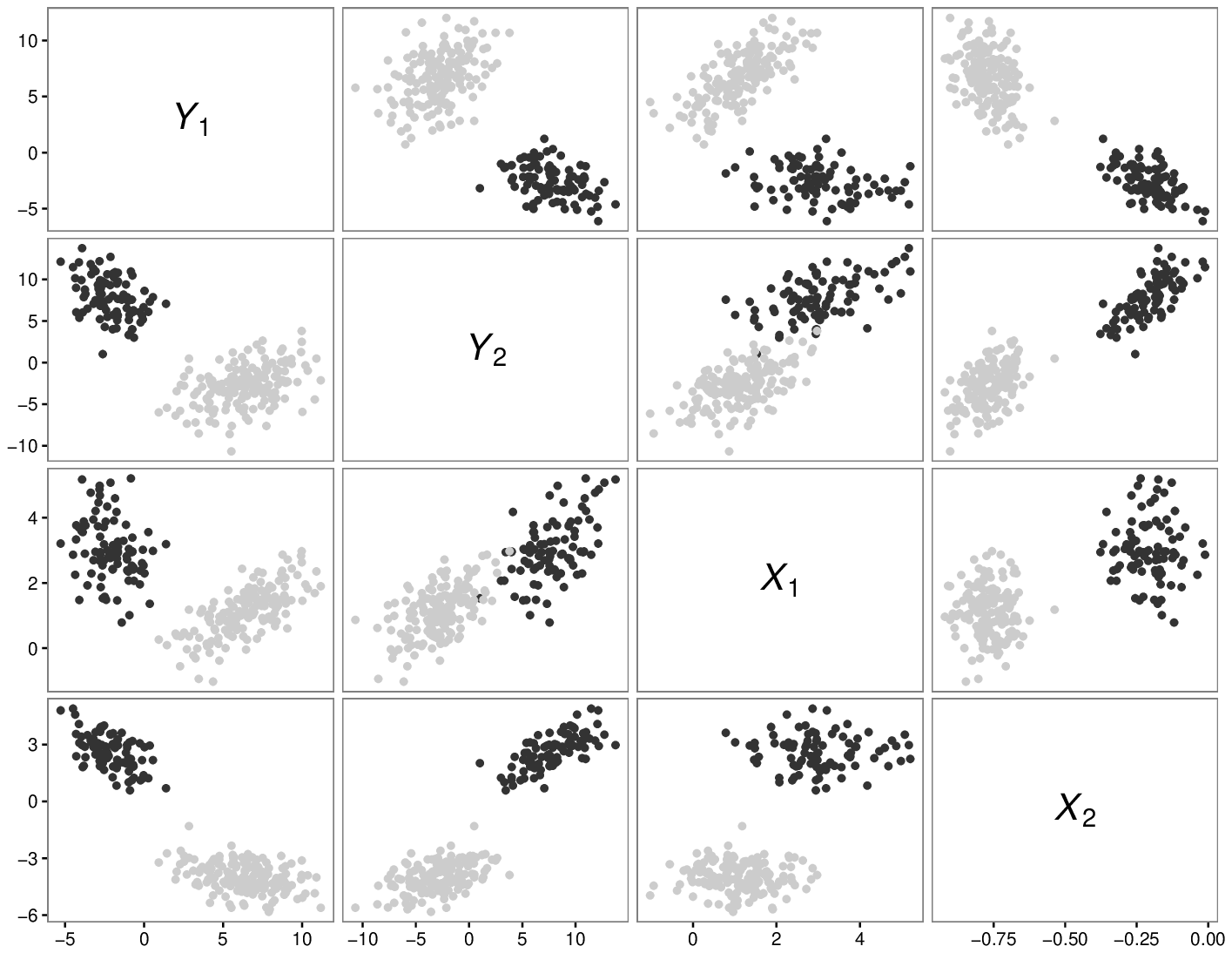}
}
\vspace{-0.15in}
\caption{Scatter plots showing an example of a generated data set for Simulation 1.}
\label{simscat}
\end{center}
\end{figure}

One hundred samples are generated. The eMCWM family is run using all 196 models for $G\in\{1,\ldots,4\}$, resulting in a total of 784 models for each sample. Both the BIC and the ICL chose the same model each time. The VEE-VII model is chosen 91 out of 100 times.  
The selected model fit the correct number of components and resulted in perfect classification 99 out of 100 times (with one misclassification on the remaining sample). 
True and estimated parameters for the two-component VEE-VII model are reported in \tablename~\ref{parestsim1}. Clearly, the parameter estimates are close to the true values.
\begin{table*}[!ht]
\caption{True parameter values along with mean and standard deviations of the parameter estimates (rounded off to two decimals) for the VEE-VII model from the 100 runs for Simulation 1.} \label{parestsim1}
\begin{tabular*}{1.0\textwidth}{@{\extracolsep{\fill}}lrrr}
\hline
Parameter & True values & Mean estimates & Standard deviations\\
  \hline
$\pi_1$ & 0.35 & 0.35 & 0.03\\
$\pi_2$ & 0.65 & 0.65 & 0.03\\ 
$\bmu_{\bX1}$ & $(3, 2.5)'$ & $(2.99, 2.51)'$ & $(0.10, 0.13)'$\\
$\bmu_{\bX2}$ & $(1.1, -4)'$ & $(1.09, -4.00)'$ & $(0.05, 0.05)'$\\
$\lambda_{\bX1}$ & 1 & 0.98 & 0.11\\
  $\lambda_{\bX2}$ & 0.5 & 0.50 & 0.04 \\
$\bbeta'_1$ & $\begin{pmatrix}
2 & -0.5 & -1\\ 
-2 & 1.5 & 2
\end{pmatrix}$ & $\begin{pmatrix}
1.97 & -0.49 & -1.00\\ 
-1.99 & 1.52 & 1.97
\end{pmatrix}$ & $\begin{pmatrix}
0.49 & 0.11 & 0.12\\ 
0.58 & 0.14 & 0.13
\end{pmatrix}$\\
$\bbeta'_2$ &
$\begin{pmatrix}
0 & 2.2 & -1\\ 
1 & 2 & 1.5
\end{pmatrix}$ & $\begin{pmatrix}
-0.07 & 2.18 & -1.02\\ 
1.09 & 1.98 & 1.52
\end{pmatrix}$ & $\begin{pmatrix}
0.68 & 0.15 & 0.15\\ 
0.79 & 0.20 & 0.18
\end{pmatrix}$\\
$\bSigma_{\bY1}$ & $\begin{pmatrix}
0.92 & 0.56\\ 
0.56 & 1.40 \\ 
  \end{pmatrix}$ & $\begin{pmatrix}
0.91 & 0.54\\ 
0.54 & 1.37 \\ 
  \end{pmatrix}$ & $\begin{pmatrix}
0.12 & 0.10\\ 
0.10 & 0.18 \\ 
  \end{pmatrix}$ \\
  $\bSigma_{\bY2}$ & $\begin{pmatrix}
1.73 & 1.05\\ 
1.05 & 2.63 \\  
  \end{pmatrix}$ & $\begin{pmatrix}
1.67 & 1.00\\ 
1.00 & 2.53 \\  
  \end{pmatrix}$ & $\begin{pmatrix}
0.15 & 0.13\\ 
0.13 & 0.26 \\  
  \end{pmatrix}$\\
\hline
\end{tabular*}
\bigskip
\end{table*}

\subsubsection{Simulation 2} \label{sec:VVVVVV}

A two-component mixture is simulated with $450$ observations with two-dimensional Gaussian responses and three-dimensional Gaussian covariates. The VVV covariance structure is used to simulate the data for both the responses and the covariates (\figurename~\ref{simscatvvv}). 
The sample sizes for each group are sampled from a binomial distribution with success probability 0.40. As in Simulation 1, 100 samples are generated. The eMCWM family is again run for $G\in\{1,\ldots,4\}$. Both the BIC and the ICL chose the same model (and two components) each time. True and estimated parameters for the two-component VVV-VVV model are given in \tablename~\ref{parestsim2}. Clearly, the parameter estimates are close on average to the true parameters. 
As compared to Simulation 1, where the generating model was selected by the model selection criteria in a majority of runs, more
parsimonious models are usually picked here as opposed to the VVV-VVV model (which is selected only fifteen times out of one hundred).
The estimated ARI values for the selected models range between 0.94 and 1.00 with a median (mean) value of 0.98 (0.98) over the 100 runs.
\begin{figure}[htb]
\begin{center}
\resizebox{0.7\textwidth}{!}{
\includegraphics[angle=-90] %,origin=c
{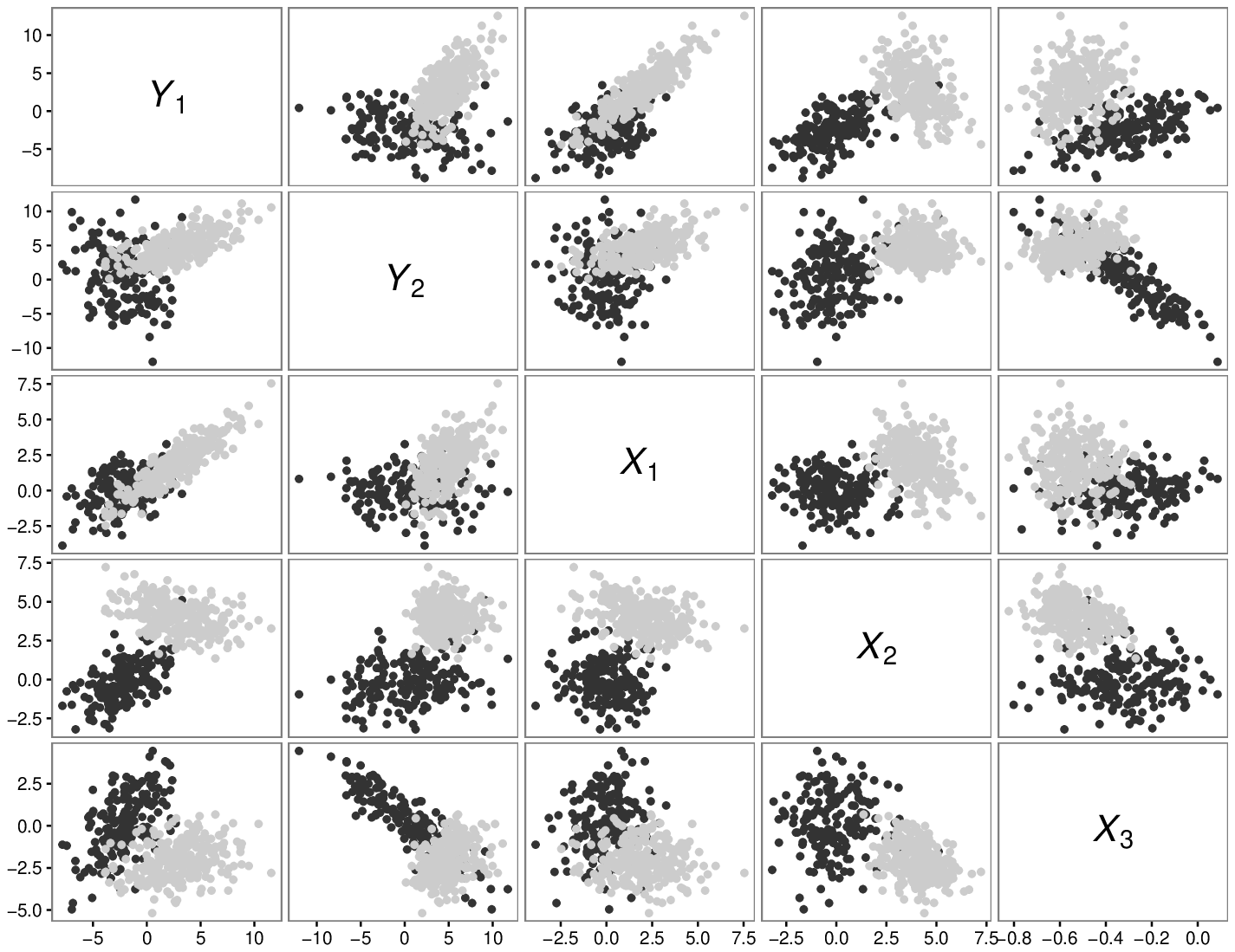}
}
\vspace{-0.15in}
\caption{Scatter plots showing an example of a generated data set for Simulation 2.}
\label{simscatvvv}
\end{center}
\end{figure}

\begin{table*}[!ht]
\caption{True parameter values along with mean and standard deviations of the parameter estimates (rounded off to two decimals) for the VVV-VVV model from the 100 runs for Simulation 2.} \label{parestsim2}
%\resizebox*{1\textwidth}{!}{
\begin{tabular*}{1.0\textwidth}{@{\extracolsep{\fill}}lrrr}
%\begin{tabular*}{1.0\textwidth}{@{\extracolsep{\fill}}lrrr}
\hline
Parameter & True values & Mean estimates & Standard deviations\\
  \hline
$\pi_1$ & 0.40 & 0.40 & 0.02\\
$\pi_2$ & 0.60 & 0.60 & 0.02\\ 
$\bmu_{\bX1}$ & $(0, 0, 0)'$ & $(-0.00, -0.01, -0.00)'$ & $(0.09, 0.10, 0.12)'$\\
$\bmu_{\bX2}$ & $(2, 4, -2)'$ & $(2.00, 4.00, -2.00)'$ & $(0.09, 0.05, 0.06)'$\\
$\bSigma_{\bX1}$ & $\begin{pmatrix}
1.72 & -0.18 & 0.27 \\ 
-0.18 & 1.89 & 0.27 \\ 
0.27 & 0.27 & 2.89
  \end{pmatrix}$ & $\begin{pmatrix}
1.70 & -0.15 & 0.27 \\ 
-0.15 & 1.88 & 0.27 \\ 
0.27 & 0.27 & 2.90
\end{pmatrix}$ & $\begin{pmatrix}
0.19 & 0.17 & 0.16 \\ 
0.17 & 0.23 & 0.17 \\ 
0.16 & 0.17 & 0.29
  \end{pmatrix}$\\
  $\bSigma_{\bX2}$ & $\begin{pmatrix}
2.33 & -0.52 & -0.06 \\ 
-0.52 & 0.88 & -0.34 \\ 
-0.06 & -0.34 & 1.04
  \end{pmatrix}$ & $\begin{pmatrix}
2.35 & -0.52 & -0.08 \\ 
-0.52 & 0.88 & -0.34 \\ 
-0.08 & -0.34 & 1.03
  \end{pmatrix}$ & $\begin{pmatrix}
0.21 & 0.10 & 0.11 \\ 
0.10 & 0.07 & 0.06 \\ 
0.11 & 0.06 & 0.09
\end{pmatrix}$ \\
$\bbeta'_1$ & $\begin{pmatrix}
-2 & 0.75 & 1 & 0.5\\ 
1 & 0.5 & 1 & -2
\end{pmatrix}$ & $\begin{pmatrix}
-2.00 & 0.74 & 1.00 & 0.49 \\ 
 0.99 & 0.49 & 1.00 & -2.00 
\end{pmatrix}$ & $\begin{pmatrix}
0.09 & 0.07 & 0.06 & 0.06 \\ 
0.09 & 0.08 & 0.07 & 0.06
\end{pmatrix}$\\
$\bbeta'_2$ &
$\begin{pmatrix}
0.5 & 1.75 & 0.25 & 1 \\ 
 1 & 1 & 1 & 1
\end{pmatrix}$ & $\begin{pmatrix}
0.51 & 1.75 & 0.24 & 1.00 \\ 
1.03 & 1.00 & 1.00 & 1.00
\end{pmatrix}$ & $\begin{pmatrix}
0.25 & 0.03 & 0.06 & 0.05 \\ 
0.38 & 0.05 & 0.09 & 0.08
\end{pmatrix}$\\
$\bSigma_{\bY1}$ & $\begin{pmatrix}
1.34 & 0.47\\ 
0.47 & 1.66 \\ 
  \end{pmatrix}$ & $\begin{pmatrix}
1.32 & 0.47 \\ 
0.47 & 1.63  
  \end{pmatrix}$ & $\begin{pmatrix}
0.13 & 0.11 \\ 
0.11 & 0.17 
  \end{pmatrix}$ \\
  $\bSigma_{\bY2}$ & $\begin{pmatrix}
0.50 & 0.04 \\ 
0.04 & 1.50 
\end{pmatrix}$ & $\begin{pmatrix}
0.50 & 0.03 \\ 
0.03 & 1.47 
  \end{pmatrix}$ & $\begin{pmatrix}
0.04 & 0.06 \\ 
0.06 & 0.14  
  \end{pmatrix}$\\
\hline
%\end{tabular*}
\end{tabular*}
\bigskip
\end{table*}

\subsubsection{Simulation 3} \label{sec:highdim}
Parameter recovery is also illustrated on higher dimensional data generated from an EEE-EEE model. One hundred samples of a two-component mixture model are simulated with nine-dimensional covariates and ten-dimensional responses. Group sample sizes are sampled from a binomial distribution with success probability 0.3 and an overall sample size of 1000. Covariates for the first component are simulated from a nine-dimensional Gaussian distribution with zero mean. Covariates for the second component are simulated with mean $(1, 2, -1, -2, 0, 0, 1, 2, -1)'$. The ten by ten matrix of regression coefficients for the two groups is simulated using a standard Gaussian distribution and fixed for all 100 runs (see Appendix \ref{sim3betas} for the coefficient values). 
The common covariance matrix for the covariates is generated using $$\begin{pmatrix}
1.00 & 0.80 & 0.60 \\ 
0.80 & 1.20 & 0.40 \\ 
0.60 & 0.40 & 0.80
 \end{pmatrix} \otimes \boldsymbol{I}_{3},$$
where $\boldsymbol{I}_{3}$ denotes a 3-dimensional identity diagonal matrix. Similarly, the common covariance matrix for the responses is generated using $$\begin{pmatrix}
1.50 & 0.60 \\ 
0.60 & 2.00
 \end{pmatrix} \otimes \boldsymbol{I}_{5},$$
where $\boldsymbol{I}_{5}$ denotes a 5-dimensional identity diagonal matrix.
The recovered parameter estimates for the EEE-EEE model are found to be close on average to the generating parameters.  Due to space constraints and the high dimensionality, the Frobenius norms of the biases of the parameter estimates from the EEE-EEE models are reported in Table 4, similar to \cite{murray2014}. 

Note that while the purpose of this simulation is to investigate parameter estimation in higher dimensions, clustering performance is also evaluated. Two-, three-, and four-component models were selected 46, 29, and 25 times, respectively. Interestingly, selection of the three- and four-component models seem to result from spurious clusters.  Spurious clusters are well known in the mixture modelling literature \citep{ingrassia2004,ingrassia2007} and are commonly associated with low variance of a mixture component, leading to a spike in the log-likelihood \citep{mclachlan2000,ingrassia2004,ingrassia2007}. 
In this simulation, the above is true for every run where a two-component model was not selected in the current simulation by the BIC. 
Such cases have been dealt with by imposing bounds on the eigenvalues of the covariance matrix of interest during parameter estimation \citep{ingrassia2004,browne2013constrained}. This, in effect, imposes a bound on the variance along the principal axes.
Here, a simple post-hoc procedure is implemented. If either of the covariates or responses covariance matrices for a model possess any eigenvalues less than a conservative bound, $\epsilon = 10^{-20}$, that model is removed from the results. This procedure has the effect of disregarding all models with spurious clusters in our simulation. As a result, a two-component model is selected all 100 times; specifically, the EEE-EEE (EEE-VEE) model was selected 99 (1) times. The estimated ARI values for the selected models range between 0.996 and 1.00 with a median (mean) value of 1.00 (1.00).

\begin{table*}[!ht]
\caption{Frobenius norms of the biases of the parameter estimates (rounded off to two decimals) for the EEE-EEE model from the 100 runs for Simulation 3.} \label{mpeparestsim3}
\smallskip
\begin{center}
\begin{tabular}{lr}%{1.0\textwidth}{@{\extracolsep{\fill}}lr}
\hline
Parameter & $\|\text{Bias}\|$\\
  \hline
$\pi_1$ & 0.00\\
$\pi_2$ & 0.00\\ 
$\bmu_{\bX1}$ & 0.01\\
$\bmu_{\bX2}$ & 0.01\\
$\bSigma_{\bX1} = \bSigma_{\bX2}$ & 0.03\\
%$\bSigma_{\bX2}$ & 0.03\\
$\bbeta'_1$ & 0.13\\
$\bbeta'_2$ & 0.11\\
$\bSigma_{\bY1}=\bSigma_{\bY2}$ &  0.12\\
%$\bSigma_{\bY2}$ & 0.12\\
\hline
\end{tabular}
\end{center}
\bigskip
\end{table*}

\subsection{Analysis of Real Data Sets}
\label{sec:realdata}

\subsubsection{Australian Institute of Sports Data} 

The Australian Institute of Sports (AIS) data \citep{cook1994} contains measurements on 202 athletes (100 female and 102 male) and is available in the {\sf R} package {\texttt{sn}} \citep{sn}. A subset of seven variables that has recently been used in the mixtures of regression literature \citep{soffritti2011} is analyzed here: red cell count ($\mathsf{RCC}$), white cell count ($\mathsf{WCC}$), plasma ferritin concentration ($\mathsf{PFC}$), body mass index, sum of skin folds, body fat percentage, and lean body mass. The blood composition variables ($\mathsf{RCC}$, $\mathsf{WCC}$, and $\mathsf{PFC}$) are selected as the response variables with the biometrical variables being the predictors. All algorithms are run for $G\in\{1,\ldots,4\}$. Table~\ref{aistable} summarizes the results from running the eMCWM, eFMR, eFMRC, FMR, FMRC, 
 and \texttt{mixture} GPCM algorithms. 

\begin{table}[htb]
\caption{Comparison of the performance of the models applied to the AIS data.\label{aistable}}
\begin{tabular*}{\textwidth}{@{\extracolsep{\fill}}lrrrr}
\hline 
Algorithm & Model & $G$ & ARI & Parameters\\
\hline
eMCWM & VVI-VVE & 2 & 0.92 & 59\\
eFMR & VI & 1 & 0 & 18 \\
eFMRC & VI & 1 & 0 & 18\\
FMR &  & 1 & 0 & 18 \\
FMRC &  & 1 & 0 & 18 \\
\texttt{mixture} & EVE & 3 & 0.60 & 63 \\
\hline
\end{tabular*}\\
{\scriptsize*Note that, for a one-component model, the family comprises three covariance structures: VI (diagonal with different entries), EI (diagonal with same entries), and VV (full covariance matrix).}

\end{table}

Both the BIC and the ICL chose a two-component model with the VVI and VVE covariance structures for the response and covariates, respectively. This model yielded an ARI of 0.92. The estimated classification for the chosen eMCWM model is in \tablename~\ref{classais}. Note that the difference between the chosen eMCWM model and the model with the second best BIC value is small ($\approx$ 0.3 BIC points). This latter model (VEI-VVE; 57 degrees of freedom) picked two components and resulted in an ARI of 0.87. The chosen (one-component) model from the eFMR and eFMRC families did not perform well, and neither did the FMR and FMRC models. The \texttt{mixture} software, on the other hand, selected a three-component solution but the estimated classification does not reconcile as well with the known grouping.
\begin{table}[htb]
\caption{Cross-tabulation of true and estimated classifications for two methods applied to the AIS data.\label{classais}}
\begin{tabular*}{1.0\textwidth}{@{\extracolsep{\fill}}lrrrrrr}
            \hline
           & \multicolumn{2}{c}{eMCWM}&& \multicolumn{3}{c}{\texttt{mixture}}\\
          \hline %\cline{2-3}\cline{5-7}
		&1&2&&1&2&3\\
		\cline{2-3}\cline{5-7}
                  Female&99&1&&83&14&3\\
                 Male&3&99&& 0&86&16\\
			\hline
\end{tabular*}
\end{table}

\subsubsection{Iris Data} 

The algorithms are also run on the famous Iris data set \citep{anderson1935, fisher1936}. This data set, available in the \texttt{datasets} package as part of {\sf R}, provides measurements on sepal length and width as well as petal length and width for 50 flowers from each of three species of Iris: \emph{setosa}, \emph{versicolor}, and \emph{virginica}. The width measurements are taken to be the response variables with the other variables as the covariates. The algorithms are run for $G\in\{1,\ldots,4\}$ (\tablename~\ref{iristable}). The selected eMCWM model is a three-component model with an ARI of 0.90 (\tablename~\ref{classiris}). Note that the difference between the chosen eMCWM model and the model with the second best BIC value is $\approx$ 1.25 BIC points. This latter model selected a two-component model (VVV-VVV; 29 degrees of freedom) with an ARI of 0.57. This model put datapoints from \emph{versicolor} and \emph{virginica} together in one group.
The model with the second best BIC value is the model that the ICL chose: a two-component VVV model with an ARI of 0.57.
\begin{table}[htb]
\caption{Comparison of the performance of the models applied to the Iris data.\label{iristable}}
\begin{tabular*}{\textwidth}{@{\extracolsep{\fill}}lrrrr}
\hline 
Algorithm & Model & $G$ & ARI & Parameters\\
\hline
eMCWM & VEV-VEV & 3 & 0.90 & 40 \\
eFMR & VEI & 2 & 0.14 & 16 \\
eFMRC & VVI & 2 & 0.45 & 19\\
FMR &  & 2 & 0.19 & 17 \\
FMRC &  & 2 & 0.45 & 19\\
\texttt{mixture} & VEV & 2 & 0.57 & 26\\
\hline
\end{tabular*}
\end{table}

The eFMR and FMR models resulted in two-component models yielding poor ARI values. A two-component VVI model is selected from the eFMRC family with an ARI of 0.45. This model clusters \emph{setosa} and \emph{virginica} perfectly with observations from \emph{versicolor} assigned to the other two clusters. Because the \texttt{flexmix} FMRC algorithm is, in essence, a VVI model, unsurprisingly, it also chose a two-component model with an ARI of 0.45.  The GPCM family as implemented in \texttt{mixture} picked a two-component model (ARI=0.57) with the data points from \emph{versicolor} and \emph{virginica} pooled together in one group (Table~\ref{classiris}).
\begin{table}[htb]
\caption{Cross-tabulation of true and estimated classifications for three methods applied to the Iris data. \label{classiris}}
\begin{tabular*}{1.0\textwidth}{@{\extracolsep{\fill}}lrrrrrrrrr}
            \hline
           & \multicolumn{3}{c}{eMCWM}&& \multicolumn{2}{c}{\texttt{mixture}}&& \multicolumn{2}{c}{eFMRC}\\
           \hline
		&1&2&3&&1&2&&1&2\\
		\cline{2-4}\cline{6-7}\cline{9-10}
                 \emph{setosa}&50& & && 50& && 50&\\
                 \emph{versicolor}& &45&5&& & 50 && 31&19\\
		    \emph{virginica}& & &50&& &50 && &50\\
			\hline
\end{tabular*}
\end{table}

\subsubsection{Crabs Data}

The crabs data set contains five morphological measurements on each of 50 crabs representing both sexes and colours (blue and orange) of the species \emph{Leptograpsus variegatus}. These data were originally introduced in \citet{campbell1974} and are available as part of the {\tt{MASS}} package \citep{venables2002} for {\sf R}. The data are famous for having highly correlated measurements on width of frontal region just anterior to frontal tebercles ($\mathsf{FL}$), width of posterior region ($\mathsf{RW}$), carapace length ($\mathsf{CL}$), carapace width ($\mathsf{CW}$), and body depth ($\mathsf{BD}$). The variables $\mathsf{CW}$, $\mathsf{FL}$, and $\mathsf{RW}$ reflect width measurements and are taken to be the response variables, with $\mathsf{CL}$ and $\mathsf{BD}$ as the predictor variables. Based on the two binary variables, sex and colour, there are four known classes within these data. The algorithms are run for $G\in\{1,\ldots,9\}$ and the results are summarized in Table~\ref{crabstable}. The selected eMCWM model is a four-component model with an ARI of 0.82 (Table~\ref{classcrabs}). Note that the difference between the chosen eMCWM model and the model with the second best BIC value is $\approx$ 0.6 BIC points. This latter model (EEE-EEE; 56 degrees of freedom) also picked a four-component model, with an ARI of 0.84.
\begin{table}[htb]
\caption{Comparison of the performance of the models applied to the crabs data.\label{crabstable}}
\begin{tabular*}{\textwidth}{@{\extracolsep{\fill}}lrrrr}
\hline 
Algorithm & Model & $G$ & ARI & Parameters\\
\hline
eMCWM & EEE-EVE & 4 & 0.82 & 59 \\
eFMR & VVI & 2 & 0.40 & 25 \\
eFMRC & EEE & 4 & 0.83 & 51\\
FMR &  & 2 & 0.40 & 25 \\
FMRC &  & 3 & 0.69 & 42\\
\texttt{mixture} & EEV & 4 & 0.78 & 68\\
\hline
\end{tabular*}
\end{table}

The chosen eFMR model is a two-component VVI model with an ARI of 0.40. 
Because the VVI model assumes independence between the response variables, it is equivalent to the \texttt{flexmix} FMR model; thus, unsurprisingly, the chosen FMR model is a two-component model with an ARI of 0.40 (Table~\ref{classcrabs}).
Note that the estimated classification from the selected two-component eFMR model leads to good separation between the sexes of the crabs. If the class membership agreement is estimated based only on the sexes of the crabs, an ARI of 0.81 is achieved. The selected FMRC model fit three components (Table~\ref{classcrabs}) with an ARI of 0.69. It basically pools together the blue males and blue females while putting the orange males and females in different clusters. eFMRC picked a four-component model with an ARI of 0.83. \texttt{mixture} picked a four-component model (Table~\ref{classcrabs}) with an ARI of 0.78. Note that, even though the ARIs achieved from the selected eMCWM, eFMRC, and \texttt{mixture} models are close, the \texttt{mixture} model estimates more parameters than the models that utilize linear dependencies between variables. Also, note that the performance of the eFMRC model should not be surprising. Recall that, in an FMRC model, $\pi_g(\bx)$ are modelled by a multinomial logit model \citep{desarbo1988}. \citet{anderson1972} noted that this multinomial condition is satisfied if the covariate densities $p(\bx)$ are assumed to be multivariate Gaussian with the same covariance matrices, i.e., the EEE model for the covariates \citep[cf.][]{ingrassia2012}. Hence, the eFMRC EEE model should give similar clustering results to the eMCWM EEE-EEE model. As pointed out above, these four-component models yield ARI values of 0.83 and 0.84, respectively.
\begin{table}[htb]
\caption{Cross-tabulation of true and estimated classifications for four methods applied to the crabs data.\label{classcrabs}}
\begin{tabular*}{1.0\textwidth}{@{\extracolsep{\fill}}lrrrrrrrrrrrrrrrrr}
            \hline
           & \multicolumn{4}{c}{eMCWM}&& \multicolumn{4}{c}{eFMRC} && \multicolumn{4}{c}{\texttt{mixture}} && \multicolumn{2}{c}{eFMR}\\
      \hline %    \cline{2-5}\cline{7-10}\cline{12-15}\cline{17-18}
		&1&2&3&4&&1&2&3&4&&1&2&3&4&&1&2\\
		\cline{2-5}\cline{7-10}\cline{12-15}\cline{17-18}
                 BM&39&11& & && 39 &11& & &&38&12& & &&46&4\\
                 BF& &50& & &&  &50 & & && &49& &1&&4&46\\
                 OM& & &50& && & &50& && & &50& &&50& \\
                 OF& & &4&46&& & &3&47&& & &5&45&&2&48\\
			\hline
\end{tabular*}\\
{\scriptsize ``B'', ``O'', ``M'', and ``F'' refer to blue, orange, male, and female, respectively.}
\end{table}

\section{Discussion}
\label{sec:discussion}

A novel family of cluster-weighted models called the eMCWM family is presented. These models can account for heterogeneous regression data with multivariate correlated responses.
The distribution of the covariates is also explicitly incorporated in the likelihood, which to the authors' knowledge is also novel in multivariate response regression methodologies. 
This allows for separate imposition of an eigen-decomposed structure on the component covariance matrices of both the responses and the covariates. Hence, eMCWM can handle data where $\bX$ and $\bY|\bx$ might have different covariance structures. For this family, identification conditions are also provided.

The eMCWM family is parsimonious. Note that the completely unconstrained GPCM model (VVV) fits the same number of parameters as the completely unconstrained eMCWM model (VVV-VVV).  %For the iris data, the ICL picked a 2 component VVV-VVV model. 
The family's performance is investigated on simulated and real benchmark data, where  
%When the eMCWM is applied to simulated and real data, 
the BIC and the ICL are in agreement with each other for all but the iris data. 
More generally, the eMCWM performed better in comparison to the FMR and FMRC models because it explicitly uses the distribution of the covariates. Not using that information led to estimated clusters (from the eFMR and eFMRC as well as the {\tt flexmix} FMR and FMRC models) that did not agree with the observed grouping of the data. In comparison to the GPCMs, taking into account the regression structure by use of linear dependencies aided in better clustering performance on some benchmark data sets. As implemented in \texttt{flexmix}, FMR only models the distribution of the (assumed independent) $\bY|\bx$, while FMRC models both the distribution of (assumed independent) $\bY|\bx$ and a logistic model of the covariates, respectively. GPCMs and their analogues for other distributions \citep[e.g.,][]{andrews2012,Punz:McNi:Robu:2013,vrbik14,dang2015b} do not account for linear dependencies and rely on modelling the data directly with an appropriate distribution. However, the eMCWM family models both linear dependencies and the distribution of the covariates. This results in better clustering performance when a clear regression relationship exists between the variables. Numerically, the EM algorithm is stable. However, to prevent fitting small components with low generalized variance (determinant of the covariance matrix) for the eMCWM family, the component sizes are computed before each M-step and a preset minimum size of the clusters is used \citep[cf.][]{mclachlan2000,celeux1988,flexmix}.

The framework presented lends itself to a straightforward extension to model-based classification \citep[e.g.,][]{mcnicholas10b,andrews11b} and discriminant analysis \citep{hastie96}. In the case of the eMCWM family, each response variable is currently regressed individually on a common set of predictor variables. This can be extended to take advantage of correlations between the response variables to improve predictive accuracy, in the fashion of the `curds and whey' method \citep{breiman1997}. Here, the Gaussian distribution is used for both the distribution of the covariates and the response for the eMCWM family. For heavier tailed data, mixtures of more robust distributions, such as the multivariate $t$ distribution \citep{andrews2012} or the multivariate power exponential distribution (that can also model lighter tailed data) \citep{dang2015b} may be employed. The use of continuous distributions may be restrictive and more work needs to be done to incorporate mixed type data for both responses and covariates \citep[see][for the univariate case $d=p=1$]{Punz:Ingr:Clus:2015}.

%\section*{Acknowledgements}
%This work is supported by an Alexander Graham Bell Canada Graduate Scholarship and Discovery Grant from the Natural Sciences and Engineering Research Council (NSERC) of Canada.

\vspace*{1pc}

\noindent {\bf{Conflict of Interest}}

\noindent {The authors have declared no conflict of interest.}

\appendix
\section*{Appendix}

\section{Proof of Theorem~\ref{thm:identif}}
\label{subsec:identifiability}

\begin{proof} 
The proof builds upon results given in \citet{Henn:Iden:2000} and \citet{Ingr:Punz:Vitt:TheG:2014}.  
Consider the class of models defined in \eqref{eq:class_gencwm}. The equality
\begin{equation}
\sum^{G}_{g=1} \phi_d(\by|\bx, \bmu_{\bY}(\bx|\bbeta_g),\bSigma_{\bY g}) 
\phi_p(\bx|\bmu_{\bX g},\bSigma_{\bX g}) \pi_g
=
\sum^{\widetilde{G}}_{s=1} \phi_d(\by|\bx, \bmu_{\bY}(\bx|\widetilde{\bbeta}_s),\widetilde{\bSigma}_{\bY s}) 
\phi_p(\bx| \widetilde{\bmu}_{\bX s},\widetilde{\bSigma}_{\bX s})\widetilde{\pi}_s \label{eq:identif_cwm}
\end{equation}
will be proven to hold for almost all $\bx \in \mR^p$ and for all $\by \in \mR^d$ %if and only if
%$G = \widetilde{G}$ and for each $g \in \{1, \ldots, G\}$ there exists $s\in\{1,\ldots,G
%\}$ such that $\bbeta_g = \widetilde{\bbeta}_s$, $\lambda_g = \widetilde{\lambda}_s$, $\bmu_g=\widetilde{\bmu}_s$, 
%$\bSigma_g = \widetilde{\bSigma}_s$ and $\pi_g = \widetilde{\pi}_s$.
if and only if
$G = \widetilde{G}$, and for each $g \in \{1, \ldots, G\}$, there exists $s\in\{1,\ldots,G
\}$ such that $\bbeta_g = \widetilde{\bbeta}_s$,  $\bmu_{\bX g}=\widetilde{\bmu}_{\bX s}$, 
$\bSigma_{\bX g} = \widetilde{\bSigma}_{\bX s}$, $\bmu_{\bY g}=\widetilde{\bmu}_{\bY s}$, 
$\bSigma_{\bY g} = \widetilde{\bSigma}_{\bY s}$, and $\pi_g = \widetilde{\pi}_s$.

Integrating each side of \eqref{eq:identif_cwm} over $\mR^d$ yields
\begin{equation}
\sum^{G}_{g=1}  \phi_p(\bx|\bmu_{\bX g},\bSigma_{\bX g}) \pi_g= \sum^{\widetilde{G}}_{s=1}  \phi_p(\bx|\widetilde{\bmu}_{\bX s},\widetilde{\bSigma}_{\bX s}) \widetilde{\pi}_s.
\label{eq:identif_margin_cwm}
\end{equation}
Let
$$
p(\bx|\bmu_{\bX},\bSigma_{\bX},\bpi)= \sum^{G}_{g=1}  \phi_p(\bx|\bmu_{\bX g},\bSigma_{\bX g}) \pi_g
$$
and
$$
p(\bx|\widetilde{\bmu}_{\bX},\widetilde{\bSigma}_{\bX},\widetilde{\bpi})= \sum^{\widetilde{G}}_{s=1}  \phi_p(\bx|\widetilde{\bmu}_{\bX s},\widetilde{\bSigma}_{\bX s}) \widetilde{\pi}_s,
$$ 
where $\bmu_{\bX}=\left\{\bmu_{\bX g}; \, g=1,\ldots,G\right\}$, $\bSigma_{\bX}=\left\{\bSigma_{\bX g}; \,  g=1,\ldots,G\right\}$, and $\bpi=\{\pi_g; \,  g=1,\ldots,G \}$. Analogous notation applies for $\widetilde{\bmu}_{\bX},\widetilde{\bSigma}_{\bX}$, and $\widetilde{\bpi}$.
Based on Bayes' theorem, 
\begin{equation}
 p(\Omega_g|\bx,\bmu_{\bX},\bSigma_{\bX},\bpi) =  \frac{ \phi_p(\bx|\bmu_{\bX g},\bSigma_{\bX g}) \pi_g }{p(\bx|\bmu_{\bX},\bSigma_{\bX},\bpi)},
 \label{P(Omega_g)|x_genCWM}
\end{equation}
for $g=1,\ldots,G$. 
Then, model~\eqref{mGCWM} can be rewritten as
\begin{align}
p(\bx,\by|\bvartheta) & =p(\bx|\bmu_{\bX},\bSigma_{\bX},\bpi) \sum^{G}_{g=1} \phi_d(\by|\bx, \bmu_{\bY}(\bx|\bbeta_g),\bSigma_{\bY g}) p(\Omega_g|\bx,\bmu_{\bX},\bSigma_{\bX},\bpi) 
\nonumber \\ & =
p(\bx|\bmu_{\bX},\bSigma_{\bX},\bpi) p(\by|\bx,\bvartheta),
\label{eq:mGCWM_2}
\end{align}
where
\begin{equation}
p(\by|\bx,\bvartheta)= \sum^{G}_{g=1}  \phi_d(\by|\bx,  \bmu_{\bY}(\bx|\bbeta_g),\bSigma_{\bY g})p(\Omega_g|\bx,\bmu_{\bX},\bSigma_{\bX},\bpi), \quad \by \in \mR^d.
\label{eq:mGCWM_reduc}
\end{equation}

Now, the class of models defined by \eqref{eq:mGCWM_reduc} for almost all $\bx\in \mR^p$, if the equality
\begin{equation*}
\sum^{G}_{g=1}  \phi_d(\by|\bx, \bmu_{\bY}(\bx|\bbeta_g),\bSigma_{\bY g})p(\Omega_g|\bx,\bmu_{\bX},\bSigma_{\bX},\bpi) 
= 
\sum^{\widetilde{G}}_{s=1} \phi_d(\by|\bx,\bmu_{\bY}(\bx|\widetilde{\bbeta}_s),\widetilde{\bSigma}_{\bY s}) p(\Omega_s|\bx,\widetilde{\bmu}_{\bX},\widetilde{\bSigma}_{\bX},\widetilde{\bpi})
\end{equation*}
implies $G = \widetilde{G}$, and for each $g \in \{1, \ldots, G\}$, there exists $s\in\{1,\ldots,G\}$ such that $\bbeta_g = \widetilde{\bbeta}_s$, $\bSigma_{\bY g} = \widetilde{\bSigma}_{\bY s}$, $\bmu_{\bX g}=\widetilde{\bmu}_{\bX s}$, $\bSigma_{\bX g} = \widetilde{\bSigma}_{\bX s}$, and $\pi_g=\widetilde{\pi}_s$.

Recall from Section \ref{sec:methodology} that the expected value $\bmu_{\bY}$ of $Y|\Omega_g$ is related to the covariates $\bX$ through the relation $\bmu_{\bY} = \bbeta'_g \bx^*$, $g =1, \ldots, G$. %$\bmu_{\bY}\left(\bx|\bbeta_g\right) = \bbeta'_g \bx^*$ 
Let
\begin{align*}
\cX & =  \Bigl\{ \bx \in \mR^p: \mbox{ for each }g,j \in \{1, \ldots, G \}, \text{ and } s,t \in \{1, \ldots, \widetilde{G} \}: \Bigr. \nonumber \\ 
& \qquad \Bigl. \bbeta'_g \bx^*=\bbeta'_j\bx^*  
\ \Rightarrow\  
\bbeta_g=\bbeta_j, \Bigr. \nonumber  \\
& \qquad \Bigl. \bbeta'_g \bx^*=\widetilde{\bbeta}'_s\bx^*  
\ \Rightarrow\  
\bbeta_g=\widetilde{\bbeta}_s , \Bigr. \nonumber \\
& \qquad \Bigl. \widetilde{\bbeta}'_s\bx^*=\widetilde{\bbeta}'_t\bx^*  
\ \Rightarrow\   
\widetilde{\bbeta}_s=\widetilde{\bbeta}_t\Bigr\} .
\end{align*}
According to \eqref{eq:class_gencwm}, $(\bbeta_g,\bSigma_{\bY g}) \neq (\bbeta_j,\bSigma_{\bY j})$, $g \neq j$; thus, it follows that the quantities $(\bbeta'_g \bx^*, \bSigma_{\bY g})$, $g=1,\ldots, G$, are pairwise distinct for all $\bx \in \cX$ (indeed, the complement of $\cX$, i.e., $\mR^p \setminus \cX$, is formed by a finite set of hyperplanes of $\mR^p$ and, thus, $\mR^p \setminus \cX$ has null measure).

For any fixed $\bx \in \cX$,  according to \eqref{P(Omega_g)|x_genCWM}, 
%\begin{equation*}
$\{p(\Omega_1|\bx,\bmu_{\bX},\bSigma_{\bX},\bpi), \ldots, p(\Omega_G|\bx,\bmu_{\bX},\bSigma_{\bX},\bpi)\}$
%\end{equation*}
and %\linebreak 
%\begin{equation*}
$\{p(\Omega_1|\bx,\widetilde{\bmu}_{\bX},\widetilde{\bSigma}_{\bX},\widetilde{\bpi}), \ldots, p(\Omega_{\widetilde{G}}|\bx,\widetilde{\bmu}_{\bX},\widetilde{\bSigma}_{\bX},\widetilde{\bpi}) \}$
%\end{equation*}
are sets of positive numbers summing to one.
It follows that, for each $\bx \in \cX$, the density
%\begin{align*}
$p(\by|\bx,\bvartheta)$ given in \eqref{eq:mGCWM_reduc}
is a mixture of distributions of  kind \eqref{eq:mixtglm_alpha} and then it is identifiable, due to the assumptions of the theorem.
Thus, $G = \widetilde{G}$ and there exists $s\in\{1,\ldots,G\}$ such that
\begin{equation}  
\bbeta_g = \widetilde{\bbeta}_s, \quad \bSigma_{\bY g} = \widetilde{\bSigma}_{\bY s} \quad \text{and} \quad p(\Omega_g|\bx,\bmu_{\bX},\bSigma_{\bX},\bpi)=p(\Omega_s|\bx,\widetilde{\bmu}_{\bX},\widetilde{\bSigma}_{\bX},\widetilde{\bpi}).
\label{eq:eq_beta_lambda}
\end{equation}
Moreover, because $p(\Omega_g|\bx,\bmu_{\bX},\bSigma_{\bX},\bpi)$ and $p(\Omega_s|\bx,\widetilde{\bmu}_{\bX},\widetilde{\bSigma}_{\bX},\widetilde{\bpi})$ %, for $g=1, \ldots, G$ and $g=1, \ldots, \widetilde{G}$  
are defined according to \eqref{P(Omega_g)|x_genCWM}, from  \eqref{eq:eq_beta_lambda} and \eqref{eq:identif_margin_cwm}, we get:
\begin{align*}
\pi_g & = \int_{\cX} \pi_g \phi_p(\bx|\bmu_{\bX g},\bSigma_{\bX g}) d \bx \\
&  = \int_{\cX} \frac{\pi_g \phi_p(\bx|\bmu_{\bX g},\bSigma_{\bX g})}{\sum^{G}_{g=1}  \phi_p(\bx|\bmu_{\bX g},\bSigma_{\bX g}) \pi_g} \left( \sum^{G}_{g=1}  \phi_p(\bx|\bmu_{\bX g},\bSigma_{\bX g}) \pi_g  \right) d \bx \\
&  = \int_{\cX} p(\Omega_g|\bx,\bmu_{\bX},\bSigma_{\bX},\bpi)   \left( \sum^{\widetilde{G}}_{s=1}  \phi_p(\bx|\widetilde{\bmu}_{\bX s},\widetilde{\bSigma}_s) \widetilde{\pi}_s \right) d \bx \\
&  = \int_{\cX} p(\Omega_s|\bx,\bmu_{\bX},\bSigma_{\bX},\bpi)   \left( \sum^{\widetilde{G}}_{s=1}  \phi_p(\bx|\widetilde{\bmu}_{\bX s},\widetilde{\bSigma}_{\bX s}) \widetilde{\pi}_s \right) d \bx \\
&  = \int_{\cX} \frac{\widetilde{\pi}_s \phi_p(\bx|\widetilde{\bmu}_{\bX s}, \widetilde{\bSigma}_{\bX s}) }{ \sum^{\widetilde{G}}_{t=1}  \phi_p(\bx|\widetilde{\bmu}_{\bX t}, \widetilde{\bSigma}_{\bX t}) \widetilde{\pi}_t} \left( \sum^{\widetilde{G}}_{s=1}  \phi_p(\bx|\widetilde{\bmu}_{\bX s},\widetilde{\bSigma}_{\bX s}) \widetilde{\pi}_s \right) d \bx \\
&  = \int_{\cX} \widetilde{\pi}_s \phi_p(\bx|\widetilde{\bmu}_{\bX s}, \widetilde{\bSigma}_{\bX s}) d \bx = \widetilde{\pi}_s . 
\end{align*}
%Thus, for the same  pair $(g,s)$ in \eqref{eq:eq_beta_lambda}, it results $  \pi_g= \widetilde{\pi}_{s}$. 
Moreover, 
\begin{align*}
\phi_p(\bx|\bmu_{\bX g},\bSigma_g) & = \frac{p(\Omega_g|\bx,\bmu_{\bX g},\bSigma_{\bX g},\bpi) }{\pi_g}
\sum^{G}_{g=1}  \phi_p(\bx|\bmu_{\bX g},\bSigma_{\bX g}) \pi_g \\
& = \frac{p(\Omega_s|\bx,\widetilde{\bmu}_{\bX},\widetilde{\bSigma}_{\bX},\widetilde{\bpi})}{\widetilde{\pi}_s } 
 \sum^{\widetilde{G}}_{s=1}  \phi_p(\bx|\widetilde{\bmu}_{\bX s},\widetilde{\bSigma}_{\bX s}) \widetilde{\pi}_s
 = \phi_p(\bx|\widetilde{\bmu}_{\bX s}, \widetilde{\bSigma}_{\bX s}).
\end{align*}
From the identifiability of Gaussian distributions, again for the same  pair $(g,s)$ in \eqref{eq:eq_beta_lambda}, it follows that
\begin{equation*}
  \quad \bmu_{\bX g}=\widetilde{\bmu}_{\bX s} \quad \text{and} \quad \bSigma_{\bX g}=\widetilde{\bSigma}_{\bX s}, \label{eq:identif_G_Gauss}
\end{equation*}
and this completes the proof. 
\end{proof}

\section{M-step} 
\label{mstepderivation}

\paragraph{Derivation of $\hat{\bbeta}_{g}^{(k+1)}$:} For the estimate of the regression coefficients $\hat{\bbeta}_{g}^{(k+1)}$, $g=1,\ldots,G$:
\begin{equation*}
\sum_{i=1}^N  \sum_{g=1}^G \hat{z}_{ig}^{(k)} \frac{\partial Q_{1}\left(\bbeta_g,\bSigma_{\bY g}|\bvartheta^{(k)}\right)} {\partial \bbeta'_{g}}  = \boldsymbol{0}',
\end{equation*}
which implies
\begin{equation*}
\frac{\partial \left\{ \displaystyle\sum_{i=1}^N  \sum_{g=1}^G \frac{-\hat{z}_{ig}^{(k)}}{2} \left[ \left(\by_i-\bbeta'_g \bx_i^*\right)' \bSigma_{\bY g}^{-1} \left(\by_i-\bbeta'_g \bx_i^*\right)\right] \right\}}{\partial \bbeta'_{g}}  = \boldsymbol{0}' ,
\end{equation*}
yielding
\begin{equation*}
\frac{\partial \left[\displaystyle\sum_{i=1}^N  \sum_{g=1}^G \frac{-\hat{z}_{ig}^{(k)}}{2} \left( -\by'_i\bSigma_{\bY g}^{-1}\bbeta'_{g}\bx_i^* - \bx_i^{*'} \bbeta_{g}\bSigma_{\bY g}^{-1}\by_i + \bx_i^{*'} \bbeta_{g}\bSigma_{\bY g}^{-1}\bbeta'_{g}\bx_i^* \right) \right]}{\partial \bbeta'_{g}}  = \boldsymbol{0}'. 
\end{equation*}
Using properties of trace and transpose, we get
\begin{align*}
\frac{\partial \left\{\displaystyle\sum_{i=1}^N  \sum_{g=1}^G \frac{ \hat{z}_{ig}^{(k)}}{2} \left[ \mbox{tr}\left(\by'_i\bSigma_{\bY g}^{-1}\bbeta'_{g}\bx_i^*\right) + \mbox{tr}\left( \bx_i^{*'} \bbeta_{g}\bSigma_{\bY g}^{-1}\by_i \right) - \mbox{tr}\left( \bx_i^{*'} \bbeta_{g}\bSigma_{\bY g}^{-1}\bbeta'_{g}\bx_i^* \right) \right] \right\}}{\partial \bbeta'_{g}}  & = \boldsymbol{0}' \\
\frac{\partial \left\{\displaystyle\sum_{i=1}^N  \sum_{g=1}^G \frac{\hat{z}_{ig}^{(k)}}{2} \left[ \mbox{tr}\left(\bbeta'_{g}\bx_i^*\by'_i\bSigma_{\bY g}^{-1}\right) + \mbox{tr}\left( \left( \bSigma_{\bY g}^{-1}\by_i \bx_i^{*'}\right)' \bbeta'_{g} \right) - \mbox{tr}\left( \bbeta'_{g}\bx_i^* \bx_i^{*'} \bbeta_{g}\bSigma_{\bY g}^{-1} \right) \right] \right\}}{\partial \bbeta'_{g}}  & = \boldsymbol{0}'. 
\end{align*}
Taking the derivative, we obtain
\begin{align*}
\sum_{i=1}^N \frac{\hat{z}_{ig}^{(k)}}{2}\left\{\bSigma_{\bY g}^{-1} \by_i \bx_i^{*'}  + \bSigma_{\bY g}^{-1}\by_i \bx_i^{*'} - \left[\left(\bSigma_{\bY g}^{-1}\right)'\bbeta'_{g} \bx_i^* \bx_i^{*'} + \bSigma_{\bY g}^{-1}\bbeta'_{g} \bx_i^* \bx_i^{*'}\right] \right\} &= \boldsymbol{0}', 
\end{align*}
and finally
\begin{align*}
\hat{\bbeta}_{g}^{(k+1)'} &= \left(\sum_{i=1}^N \hat{z}_{ig}^{(k)}  \by_i \bx_i^{*'}\right) \left({\sum_{i=1}^N \hat{z}_{ig}^{(k)}  \bx_i^* \bx_i^{*'}}\right)^{-1}. 
\end{align*}
\paragraph{Derivation of $\hat{\bSigma}_{\bY g}^{(k+1)}$:} For the estimate of the covariance matrix $\hat{\bSigma}_{\bY g}^{(k+1)}$, $g=1, \ldots,G$:
\begin{equation*}
\sum_{i=1}^N \hat{z}_{ig}^{(k)} \frac{\partial Q_{1}\left(\bbeta^{(k+1)}_g,\bSigma_{\bY g}|\bvartheta^{(k)}\right)} {\partial \bSigma^{-1}_{\bY g}}=\boldsymbol{0}',
\end{equation*}
leading to 
\begin{align*}
\frac{\partial \left\{\displaystyle\sum_{i=1}^N \sum_{g=1}^G \frac{\hat{z}_{ig}^{(k)}}{2} \left[ \log |\bSigma_{\bY g}^{-1}| - \mbox{tr}\left( \left(\by_i-\bbeta^{(k+1)'}_g \bx_i^*\right)' \bSigma_{\bY g}^{-1} \left(\by_i-\bbeta^{(k+1)'}_g \bx_i^*\right) \right)\right] \right\}}{\partial \bSigma^{-1}_{\bY g}} &= \boldsymbol{0}', \\
\frac{\partial \left\{\displaystyle\sum_{i=1}^N \sum_{g=1}^G \frac{\hat{z}_{ig}^{(k)}}{2} \left[ \log |\bSigma_{\bY g}^{-1}| - \mbox{tr}\left(\bSigma_{\bY g}^{-1} \left(\by_i-\bbeta^{(k+1)'}_g \bx_i^*\right) \left(\by_i-\bbeta^{(k+1)'}_g \bx_i^*\right)' \right) \right] \right\}}{\partial \bSigma^{-1}_{\bY g}} &= \boldsymbol{0}'.
\end{align*}
Taking the derivative, we get
\begin{displaymath}
\sum_{i=1}^N \frac{\hat{z}_{ig}^{(k)}}{2} \left\{ \left(\bSigma_{\bY g}^{-1}\right)^{-1'} - \left[\left(\by_i-\bbeta^{(k+1)'}_g \bx_i^*\right) \left(\by_i-\bbeta^{(k+1)'}_g \bx_i^*\right)'\right]' \right\} = \boldsymbol{0}',
\end{displaymath}
and this results in
\begin{equation*}
\hat{\bSigma}_{\bY g}^{(k+1)} = \frac{\displaystyle\sum_{i=1}^N \hat{z}_{ig}^{(k)}  \left(\by_i-\hat{\bbeta}_{g}^{(k+1)'} \bx_i^*\right)\left(\by_i-\hat{\bbeta}_{g}^{(k+1)'} \bx_i^*\right)'}{\displaystyle\sum_{i=1}^N \hat{z}_{ig}^{(k)}}. 
\end{equation*}
%Please insert appendices before the references.

\section{Regression Coefficients for Simulation 3} 
\label{sim3betas}

Regression coefficients used to generate data for group 1 for Simulation 3:
$$
\begin{pmatrix*}[r]
-0.63 & 0.18 & -0.84 & 1.60 & 0.33 & -0.82 & 0.49 & 0.74 & 0.58 & -0.31 \\ 
1.51 & 0.39 & -0.62 & -2.21 & 1.12 & -0.04 & -0.02 & 0.94 & 0.82 & 0.59 \\ 
0.92 & 0.78 & 0.07 & -1.99 & 0.62 & -0.06 & -0.16 & -1.47 & -0.48 & 0.42 \\ 
1.36 & -0.10 & 0.39 & -0.05 & -1.38 & -0.41 & -0.39 & -0.06 & 1.10 & 0.76 \\ 
-0.16 & -0.25 & 0.70 & 0.56 & -0.69 & -0.71 & 0.36 & 0.77 & -0.11 & 0.88 \\ 
0.40 & -0.61 & 0.34 & -1.13 & 1.43 & 1.98 & -0.37 & -1.04 & 0.57 & -0.14 \\ 
2.40 & -0.04 & 0.69 & 0.03 & -0.74 & 0.19 & -1.80 & 1.47 & 0.15 & 2.17 \\ 
0.48 & -0.71 & 0.61 & -0.93 & -1.25 & 0.29 & -0.44 & 0.00 & 0.07 & -0.59 \\ 
-0.57 & -0.14 & 1.18 & -1.52 & 0.59 & 0.33 & 1.06 & -0.30 & 0.37 & 0.27 \\ 
-0.54 & 1.21 & 1.16 & 0.70 & 1.59 & 0.56 & -1.28 & -0.57 & -1.22 & -0.47 \\
\end{pmatrix*}
$$

\noindent Regression coefficients used to generate data for group 2 for Simulation 3:
$$
\begin{pmatrix*}[r]
-0.62 & 0.04 & -0.91 & 0.16 & -0.65 & 1.77 & 0.72 & 0.91 & 0.38 & 1.68 \\ 
-0.64 & -0.46 & 1.43 & -0.65 & -0.21 & -0.39 & -0.32 & -0.28 & 0.49 & -0.18 \\ 
-0.51 & 1.34 & -0.21 & -0.18 & -0.10 & 0.71 & -0.07 & -0.04 & -0.68 & -0.32 \\ 
0.06 & -0.59 & 0.53 & -1.52 & 0.31 & -1.54 & -0.30 & -0.53 & -0.65 & -0.06 \\ 
-1.91 & 1.18 & -1.66 & -0.46 & -1.12 & -0.75 & 2.09 & 0.02 & -1.29 & -1.64 \\ 
0.45 & -0.02 & -0.32 & -0.93 & -1.49 & -1.08 & 1.00 & -0.62 & -1.38 & 1.87 \\ 
0.43 & -0.24 & 1.06 & 0.89 & -0.62 & 2.21 & -0.26 & -1.42 & -0.14 & 0.21 \\ 
2.31 & 0.11 & 0.46 & -0.08 & -0.33 & -0.03 & 0.79 & 2.08 & 1.03 & 1.21 \\ 
-1.23 & 0.98 & 0.22 & -1.47 & 0.52 & -0.16 & 1.46 & -0.77 & -0.43 & -0.93 \\ 
-0.18 & 0.40 & -0.73 & 0.83 & -1.21 & -1.05 & 1.44 & -1.02 & 0.41 & -0.38 \\ 
\end{pmatrix*}
$$

%\bibliographystyle{natbib} 
%\bibliography{references}

\end{document}